\newif\iflipics
\lipicstrue
\lipicsfalse

\iflipics
\documentclass[a4paper,USenglish,cleveref,autoref,thm-restate,anonymous]{lipics-v2021}
\else
\documentclass[letterpaper]{article}
\usepackage[utf8]{inputenc}
\usepackage{fullpage}

\usepackage{amsmath,amsthm}
\usepackage{amssymb}
\usepackage{amsfonts}
\usepackage{dsfont}
\usepackage{xcolor,xspace}
\usepackage{tikz}
\usepackage[hidelinks]{hyperref}
\hypersetup{
    colorlinks=true,
    linkcolor=blue,
    filecolor=magenta,      
    urlcolor=cyan,
}
\fi

\usepackage{enumitem}

\usepackage{algorithm}
\usepackage{algorithmicx}
\usepackage{algpseudocode}
\algtext*{EndWhile}
\algtext*{EndIf}
\algtext*{EndFor}
\algtext*{EndFunction}
\algtext*{EndProcedure}

\usepackage[font=small,labelfont=bf,tableposition=top]{caption}
\usepackage[font=footnotesize]{subcaption}

\usepackage[colorinlistoftodos,textsize=small,color=red!25!white,obeyFinal]{todonotes}
\newcommand{\mdnote}[1]{\todo[]{Mike's Note: #1}}

\iflipics
\else
\newtheorem{lemma}{Lemma}[section]
\newtheorem{theorem}[lemma]{Theorem}
\newtheorem{claim}[lemma]{Claim}
\newtheorem{definition}[lemma]{Definition}
\newtheorem{corollary}[lemma]{Corollary}
\fi

\newtheorem{obs}[lemma]{Observation}

\DeclareMathOperator*{\E}{\mathbb{E}}
\DeclareMathOperator*\union{\bigcup}

\newcommand{\R}{\mathbb{R}}
\newcommand{\Z}{\mathbb{Z}}

\newcommand{\calP}{{\mathcal{P}}}
\newcommand{\calQ}{{\mathcal{Q}}}

\newcommand{\poly}{\mathrm{poly}}

\newcommand{\bfe}{{\mathbf{e}}}
\newcommand{\bfr}{{\mathbf{r}}}

\newcommand{\bfE}{{\mathbf{E}}}
\newcommand{\bfV}{{\mathbf{V}}}
\newcommand{\bfT}{{\mathbf{T}}}

\newcommand{\cost}{{\mathrm{cost}}}
\newcommand{\opt}{{\mathrm{opt}}}
\newcommand{\tw}{{\mathrm{tw}}}
\newcommand{\CC}{{\mathrm{CC}}}
\newcommand{\db}{{\mathrm{db}}}

\mathchardef\mhyphen="2D

\allowdisplaybreaks

\title{Degrees and Network Design: New Problems and Approximations}

\iflipics
\author{Michael Dinitz}{Johns Hopkins University}{mdinitz@cs.jhu.edu}{}{}
\author{Guy Kortsarz}{Rutgers University, Camden}{guyk@camden.rutgers.edu}{}{}
\author{Shi Li}{University at Buffalo}{shil@buffalo.edu}{}{}


\else
\author{Michael Dinitz\thanks{Supported in part by NSF grants CCF-1909111 and CCF-2228995.}\\Johns Hopkins University\\ \texttt{mdinitz@cs.jhu.edu} \and Guy Kortsarz\\Rutgers University, Camden\\ \texttt{guyk@camden.rutgers.edu} \and Shi Li\thanks{Supported in part by NSF grant CCF-1844890.}\\University at Buffalo\\ \texttt{shil@buffalo.edu}}
\date{}
\fi

\begin{document}
\maketitle

\begin{abstract}
While much of network design focuses mostly on cost (number or weight of edges), node degrees have also played an important role.  They have traditionally either appeared as an objective,  to minimize the maximum degree (e.g., the Minimum Degree Spanning Tree problem), or as constraints which might be violated to give bicriteria approximations (e.g., the Minimum Cost Degree Bounded Spanning Tree problem).  We extend the study of degrees in network design in two ways.  First, we introduce and study a new variant of the Survivable Network Design Problem where in addition to the traditional objective of minimizing the cost of the chosen edges, we add a constraint that the $\ell_p$-norm of the  node degree vector is bounded by an input parameter.  This interpolates between the classical settings of maximum degree (the $\ell_{\infty}$-norm) and the number of edges (the $\ell_1$-degree), and has natural applications in distributed systems and VLSI design. We give  a constant  bicriteria approximation in both measures using convex  programming.
Second, we provide a polylogrithmic  bicriteria approximation for the Degree Bounded Group Steiner problem  on bounded treewidth graphs,  solving an open problem from
\cite{guy1} and \cite{GKL22}.
\end{abstract}

\section{Introduction}
The overarching theme of network design problems is to find ``inexpensive'' subgraphs that satisfy some type of connectivity constraints.  The notion of ``inexpensive'' is often either the number of edges (unweighted cost) or the sum of edge costs (weighted cost).  However, it has long been recognized that in many applications \emph{vertex degrees} matter as much (or more) than cost.  This is particularly true in the context of networking and distributed systems, where the degree of a node often corresponds to the ``load'' on that node, as well as in VLSI design. So there has been a significant amount of work on handling degrees, either instead of or in addition to cost, which has led to many seminal papers and results.  With degrees as an objective, these include the well known local search approach of F\"urer and Raghavachari~\cite{FR94} for the Minimum Degree Spanning Tree problem and the Minimum Degree Steiner Tree problem.  With degrees as a constraint, these include the iterative rounding \cite{jain1} approach of Singh and Lau~\cite{SL15} for the Minimum-Cost Bounded-Degree Spanning Tree problem, as well as many extensions (most notably to Survivable Network Design with degree bounds~\cite{LNS09}, but see~\cite{ravi1} for many other examples).

In this paper we extend the study of degrees in network design in two ways.  First, we introduce what is (to the best of our knowledge) a new class of problems. Instead of bounding the cost and individual degrees as in \cite{SL15,LNS09},  
 our objective is to obtain minimum cost while satisfying a bound on the $\ell_p$-norm of the node degree vector.  This interpolates between the maximum degree (the $\ell_{\infty}$-norm) and the total number of edges or unweighted cost (the $\ell_1$-degree).  Second, we solve  a well known  open  problem: We give a poly-logarithmic bicriteria approximation for the Group Steiner Tree problem  with degree bounds on  bounded treewidth graphs.

\subparagraph{$\ell_p$-Objective.}
While the maximum degree is often a reasonable objective, as is minimizing the total cost (either with or without degree bounds), there are many natural situations where none of these approaches are fully satisfactory.  If we simply ignore the degrees and focus on cost (weighted or unweighted), then we might end up with a solution with highly imbalanced degrees, leading to large load at particular nodes.  If we ignore costs and simply optimize the maximum degree, then we might return a solution with far more edges than are needed: if the structure of the graph forces some node to have large degree, then if we simply try to minimize the maximum degree we will not even try to make the degrees of other nodes small.  Finally, optimizing under individual degree bounds implicitly assumes that nodes ``really have'' these degree bounds, i.e., they come from some external constraint.  But this is of course not always the case: often we do not have real bounds on individual nodes, but rather a more vague desire to ``keep degrees small''.

Hence we want some way of making sure that the maximum degree is small, but also encouraging few edges.  A natural function that simultaneously accomplishes both of these goals is the $\ell_p$-norm of the degree vector, i.e., the function $\left( \sum_{v \in V} (\deg_v)^p \right)^{1/p}$ for $p \geq 1$ (and in particular for $p =2$), where $\deg_v$ is the degree of $v$ in the output subgraph.  When $p=1$ this is simply (twice) the number of edges (i.e., the unweighted cost), and when $p = \infty$ this is the maximum degree.  But for intermediate values of $p$, it discourages very large degrees (in particular the maximum degree) since $p>1$ implies that large degrees have a larger effect on the norm than smaller degrees, while still also being effected on a non-trivial way by the smaller degrees.  So we can either use the $\ell_p$-norm as an objective function, or we can use it as a constraint that is far more flexible than having simple degree constraints at every node.  

This intuition, that the $\ell_p$-norm takes into account both the maximum and the distribution simultaneously, is one reason why the $\ell_p$-norm has been an important objective function 
for combinatorial problems. For example, the Set Cover problem was studied under the $\ell_p$ norm of the vector of number of elements assigned to each set \cite{anu1}. It was also extensively studied in scheduling problems (see for example \cite{azar1, azar2, arvind1, IL23}).  To the best of our knowledge, the  $\ell_p$ norm has not been studied in the context of network design,  with the notable recent exception of \emph{graph spanners}~\cite{CDR19,CDR19-approx}, where the direct applications of spanners to distributed systems led to exactly this motivation.  Similarly, MST with  $\ell_p$ norm is very important for VLSI design, since in many such settings we are forced to use spanning trees and hence the number of edges is fixed. So minimizing the $\ell_p$ norm will likely derive a balanced degree vector which is of key importance  for these VLSI application  (see, for example,  \cite{china,vlsi1}).

Motivated by the above discussion, we introduce and give the first approximation for the Survivable Network design problem with low cost under a bound on the $\ell_p$ norm of the degree vector.  

\subparagraph{Group Steiner Tree with Degree Bounds.}
In addition to  the study of  $\ell_p$-norm problems, we also make significant progress on a known open problem: approximating Group Steiner Tree  with degree bounds on bounded treewidth graphs.  The Group Steiner Tree problem (without degree bounds) is a classical optimization problem \cite{GKR1} which has played a central role in network design.  In this problem there is a designated root node $r$, and a collection of (not necessarily disjoint) \emph{groups} of vertices.  The goal is to find a subtree which connects at least one vertex from each group to $r$, while minimizing the total cost of all edges in the subtree.   The Degree Bounded Group Steiner problem was 
first raised by 
Hajiaghayi in \cite{MH} (in the 8th Workshop on Flexible Network Design), motivated by the online version of the problem and applications to VLSI design.  In particular, while low cost is highly desirable, this cost is payed only once, while later the VLSI circuit is applied (evaluated) constantly.  Low degrees
imply that the computation of the value of the circuit can be done faster.  See a discussion of why low degrees are important for Group Steiner in \cite{guy1}.

Unfortunately, despite significant recent interest in this problem~\cite{GKL22,guy1}, progress has been elusive.  In particular, polylogarithmic bicriteria approximations were not even known for simple classes such as \emph{series-parallel} graphs, i.e., for graphs with treewidth $2$.  We go far beyond series-parallel graphs, and give results for bounded treewidth graphs.

\subsection{Our Results and Techniques}

We begin in Section~\ref{sec:lp} with a study of the $\ell_p$-Survivable Network Design problem. We are given the input graph $G = (V, E)$, with edge costs $c \in \R_{\geq 0}^E$.  
    There is a connection requirement vector $r \in \Z_{\geq 0}^{\binom{V}{2}}$, a number $p \geq 1$ and a bound $A$ on the $\ell_p$ norm of the degree vector of the output graph.  The goal of the problem is to find the minimum-cost subgraph $H$ of $G$ satisfying the following:
	\begin{itemize}
		\item {\bf (connection requirements)} for every $u, v \in V$ with $u \neq v$, there are at least $r_{u, v}$ edge disjoint paths between $u$ and $v$ in $H$, and
		\item {\bf (degree constraint)} $\left(\sum_{v \in V}d_H^p(v)\right)^{1/p} \leq A$, where $d_H(v)$ is the degree of $v$ in $H$.
	\end{itemize}

We assume the input instance is feasible; that is, there is a valid sub-graph $H$ satisfying both requirements. Let $\opt$ be the minimum cost of a valid subgraph $H$. 
The main theorem we prove for the problem is the following:
\begin{theorem}
\label{thm:network-design}
There is a (randomized) algorithm which, given an instance of \textsc{$\ell_p$-Survivable Network Design}, outputs a subgraph $H$ satisfying the connection requirements and which has the following properties.
\begin{itemize}
\item The expected cost of $H$ is at most $2 \cdot \opt$.
\item The expectation of the $\ell_p$-norm of the degree vector  is at most $2^{1/p} 5^{1-1/p} \cdot A$.
\end{itemize}
\end{theorem}

For the special case of $\ell_p$-Spanning Tree problem, where $r_{uv}=1$ for all $u,v \in V$, we improve the expected cost to at most $\opt$ (rather than $2\cdot \opt$) and the expectation of the $\ell_p$-norm of the degree vector to at most $2^{1-1/p} \cdot A$.

Our main approach is to leverage the fact that the $\ell_p$-norm is \emph{convex}.  This allows us to write a convex relaxation for the problem, which can then be solved efficiently using standard convex programming techniques.  We then round this solution using an iterative rounding approach.  Making this work requires overcoming a number of issues, possibly the trickiest of which is handling fractional degrees that are less than $1$.  Note that a fractional solution could have many nodes with very small fractional degree (e.g., $1/n$).  Due to the structure of the $\ell_p$-norm, such small values contribute far less to the $\ell_p$-norm than they ``should'' (in an integral solution).  To get around this, we actually \emph{change} the $\ell_p$-constraint in a way that acts differently for values less than $1$, while still maintaining convexity.  With this change in place, we can solve the relaxation, interpret the fractional degrees ``as if'' they are true degree bounds, and then round using existing results on iterative rounding for degree-bounded network design.

\medskip
We then move to our second problem, Group Steiner Tree with Degree Bounds on bounded treewidth graphs.  In the problem, we are given a graph $G = (V, E)$ with treewidth $\tw$, a cost vector $c \in \R_{\geq 0}^E$, a root $r$, and $k$ sets $S_1, S_2, \cdots, S_k$.  We are additionally given a degree bound $\db_v \in \Z_{>0}$ for every $v \in V$. The goal of the problem is to choose a minimum-cost subgraph $H$ of $G$ such that for every $t \in k$, $H$ contains a path from $r$ to some vertex in $S_t$, and $d_H(v) \leq \db_v$ for every $v \in V$. By minimality, the optimum $H$ is always a tree. We solve an open problem from  \cite{GKL22} and \cite{guy1} by giving a polylogarthmic bicriteria algorithm as long as the treewidth is bounded.  In particular, we prove the following theorem.
\begin{theorem} \label{thm:GST-bd-treewidth}
    There is an $n^{O(\tw \log \tw)}$-time randomized algorithm for the Group Steiner Tree with Degree Bounds problem on bounded treewidth graphs which has $O(\log^2 n)$ approximation ratio and $O(\log^2 n)$-degree violation.
\end{theorem}

In order to achieve this result, we introduce and study a ``tree labeling'' problem in Section~\ref{sec:tree-labeling}. There is a rooted full binary tree, and we need to give a label $\ell_u$ for each node $u$ in the tree from a subset $L_u$ of potential labels. For every internal node $u$ with two children $v$ and $v'$ there are some consistency constraints on the labels, which say that the triple $(\ell_u, \ell_v, \ell_{v'})$ must be from some given subset $\Gamma_u \in L_u \times L_v \times L_{v'}$. Then we have some covering constraints, each specified by a set $S$ of labels: the constraint requires that at least one node has its label in $S$.  Finally, we have many cost constraints. For each such constraint, a label is given a cost, and we require that the total cost of all labels used is at most 1.  For this problem we give a randomized algorithm that outputs a labeling that satisfies all consistency constraints, and approximately satisfies the covering and cost constraints with reasonable probability, assuming the given instance is feasible. It runs in polynomial time when the depth of the tree is $O(\log n)$ and each $L_u$ has $O(1)$-size. The main techniques of the algorithm are adaptations of the LP-rounding algorithm in \cite{GKL22} for their degree-bounded network design problem. We introduce the tree labeling problem as a host for these techniques, and adapt them for the problem.

We then show in Section~\ref{sec:reduction} that we can reduce Group Steiner Tree with Degree Bounds on bounded treewidth graphs to this tree labeling problem.  Let $\tw$ be the treewidth of the graph; it is known from \cite{Bod88} that we can assume the decomposition tree of $G$ is an $O(\log n)$-depth binary tree, with bag size $O(\tw)$. This decomposition tree will be the tree in the tree-labeling instance.  For each bag in the tree, a label will contain the set of edges we take from the bag, and some connectivity information on the vertices in the bag. We define the consistency constraints so that if they are satisfied, then the connectivity information is correct. A group being connected can be captured by a covering constraint in the tree labeling instance, and the edge cost constraint and degree constraints can be formulated as cost constraints in the instance. Using the algorithm for the tree labeling instance, we obtain a tree with small cost that satisfies degree bounds approximately, and connects a group with reasonable probability. The final output then is obtained by running the procedure many times and taking the union. 

\subsection{Other Related Work}
For the survivable network design problem without any degree constraints, the classic result of Jain \cite{jain1} gives a $2$-approximation algorithm using the iterative rounding method. 
In \cite{GKR1} an $O(\log^2 n)$ approximation is given for the Group Steiner problem on tree inputs,
and an $O(\log^3 n)$ for the Group Steiner problem (without degree constraints)
for general graphs. The approximation for trees 
is almost the  best possible, unless NP problems can be solved 
in quasi-polynomial time \cite{eran1}. \cite{GKL22} gave a bicriteria 
approximation for the Group Steiner Tree Problem with degree bounds on tree inputs, with approximation ratio $O(\log^2 n)$ and degree violation $O(\log n))$.  Both bounds are nearly optimal \cite{eran1, irit7}.
In \cite{DV} the authors gave 
an $O(\log^2 n)$-approximation 
ratio for Group Steiner problem on bounded treewidth graphs
(without degree bounds).
In \cite{guy1} an $O(\log^2 n)$ approximation is given for the Group Steiner problem  with minimum maximal degree, but without costs.

\subsection{Notation}
Given a graph $H$ and a vertex $v$ in $H$, we shall use $\delta_H(v)$ to denote the set of edges in $H$ incident to $v$, and $d_H(v) = |\delta_H(v)|$ to denote its degree.  Given a rooted tree $T$ and a vertex $v$ in $T$, we use $\Lambda_T(v)$ to denote the set of children of $v$ in $T$, and $\Lambda^*_T(v)$ to denote the set of descendants of $v$ in $T$ (including $v$ itself). When $H$ and $T$ are clear from the context, we shall omit them in the subscript.  For example, this happens when $H = G$ is the input graph.

For a real vector $z$ over some domain, and a subset $S$ of elements in the domain, we define $z(S):=\sum_{i \in S}z_i$ to denote the sum of $z$ values of elements in $S$.

\section{\texorpdfstring{$\ell_p$}{Lp}-Survivable Network Design} \label{sec:lp}
    In this section, we give our iterative rounding algorithm for $\ell_p$-survivable network design problem. Recall that we are given a graph $G = (V, E)$ with cost vector $c \in \R_{\geq 0}^E$, a connection requirement vector $r \in \Z_{\geq 0}^{{\binom{V}{2}}}$, and a bound $A$ on the $\ell_p$ norm of the degree vector.

	\begin{definition}
		We say a polytope $\calP \in [0, 1]^E$ is good if it is upward-closed \footnote{This means for every $x \in \calP$ and $x' \in [0, 1]^E$ with $x' \geq x$, we have $x' \in \calP$.} and the following holds: For every vector $x \in \{0, 1\}^E$, we have that $x \in \calP$ if and only if the graph $(V, \{e \in E: x_e = 1\})$ satisfies the connection requirements.
	\end{definition}
	Notice that the above definition does not capture the degree constraints.  This is done using the following definition.  For a real vector $B \in [1, \infty]^V$, we define $\calQ_B := \{x \in [0, 1]^E: \forall v \in V, x(\delta(v)) \leq B_v\}$ to be the set of all vectors satisfying the degree bounds defined by $B$. 
			
	\begin{definition}
		\label{defn:integral}
		Let $\alpha \geq 1$ and $\beta \geq 0$ be two real numbers and $\calP$ be a good polytope. We say $\calP$ is $(\alpha,\beta)$-integral if for every $B \in [1, \infty]^V$, every non-integral extreme point $x$ of $\calP \cap \calQ_B$ satisfies at least one of the following two properties:
		\begin{enumerate}[label=(\ref{defn:integral}\alph*), leftmargin=*]
			\item there exists an edge $e \in E$ with $1/\alpha \leq x_e < 1$,  \label{event:round}
			\item there exists a vertex $v \in V$ such that $x(\delta(v)) = B_v$ and $|\{e \in \delta(v): x_e > 0\}| \leq B_v + \beta$. \label{event:relax}
		\end{enumerate}
	\end{definition}
 
	It is well known that for Survivable Network Design there is a $(2, 3)$-integral polytope $\calP$ \cite{LNS09}. For the special case of spanning tree problem, i.e, $r \equiv 1$,  there is a $(1, 1)$-integral polytope \cite{SL15}.

We will use these polytopes in our algorithm, and will show that that their existence implies good approximation algorithms.  More formally, we prove the following theorem.

\begin{theorem} \label{thm:polytope-alg}
Assuming the existence of an $(\alpha, \beta)$-integral polytope, there is a randomized algorithm which outputs a subgraph $H$ of $G$ satisfying the connection requirements. The expected cost of $H$ is at most $\alpha \cdot \text{opt}$ and the expectation of the $p$-norm of degree vector is at most $\alpha^{1/p}(\alpha+\beta)^{1-1/p}A$; recall that $\opt$ is the value of the instance.
\end{theorem}

Note that this theorem, together with the existence of a $(2,3)$-integral polytope for the general case and a $(1,1)$-integral polytope for the spanning tree case, imply Theorem~\ref{thm:network-design}.  So we focus on proving Theorem~\ref{thm:polytope-alg}.

\subsection{The Convex Program}
Define a function $f:\R_{\geq 0} \to \R_{\geq 0}$ as follows: $f(x)  = \begin{cases}
		x & \text{if } x \in [0, 1]\\
		x^p & \text{if } x > 1
	\end{cases}$.
Figure~\eqref{fig:f} shows this function for $p = 2$.  This is a convex function for $p \geq 1$.

Let $\calP$ be an $(\alpha, \beta)$-integral polytope.  The following is our convex programming relaxation for the problem:
\begin{equation}
	\min \sum_{e \in E}c_ex_e \qquad  \text{s.t.} \qquad x \in \calP, \qquad
	\sum_{v \in V} f(x(\delta(v)))  \leq A^p.
		 \label{LP}
\end{equation}
Recall that using our notation, $x(\delta(v))$ is the sum of $x$ values of edges incident to $v$ in $G$. 
\eqref{LP} is a convex program and can be solved efficiently.  Since the indicator vector of the optimum subgraph $H$ satisfies all the constraints, the value of the convex program is at most $\text{opt}$.

We note that if we instead used the function $f(x) = x^p$ (i.e., without handling the $0 \leq x \leq 1$ case separately), we would still have a convex relaxation of our problem.  However, it is not hard to show that this relaxation has an extremely large integrality gap (even if we are allowed to violate the $\ell_p$-norm constrain by a polylogarithmic factor).  Treating $0 \leq x \leq 1$ differently from $x > 1$ is one of the key ideas in our approximation algorithm.
	\begin{figure}
		\centering
		\begin{subfigure}[b]{0.28\textwidth}
			\centering
			\includegraphics[width=0.95\textwidth]{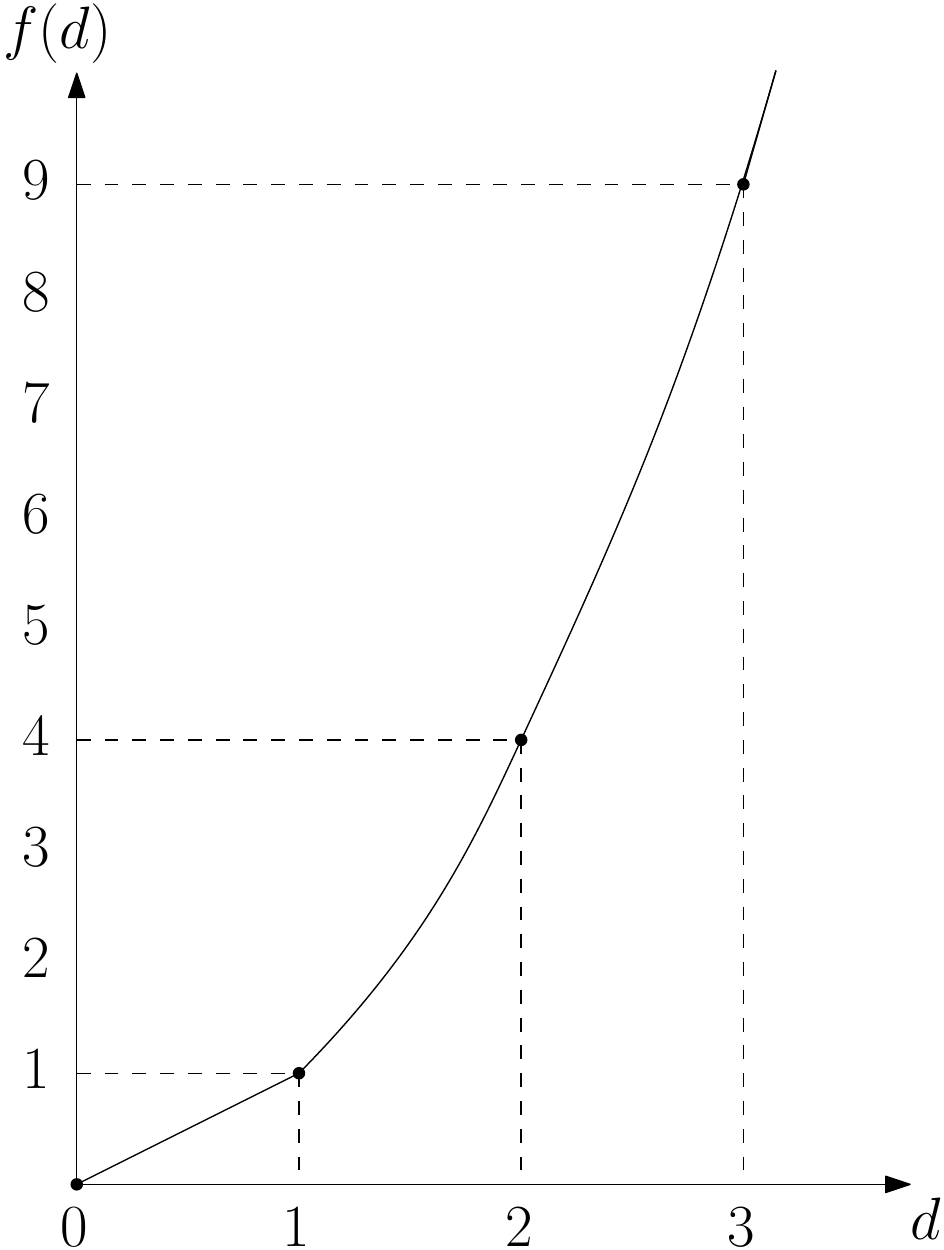}
			\caption{The function $f$ for $p = 2$.}
			\label{fig:f}
		\end{subfigure}\hspace{0.1\textwidth}
		\begin{subfigure}[b]{0.6\textwidth}
				\centering
				\begin{algorithmic}[1]
					\State Solve LP\eqref{LP} to obtain a solution $x$ \label{step:solve-LP}
					\State let $B_v \gets \max\{x(\delta(v)), 1\}$ for every $v \in V$ \label{step:init-B}
					\While{true} \label{step:while}
						\State randomly choose an extreme point $x'$ of $\calP \cap \calQ_B$ such that $\E[x'] = x$ \label{step:x'}
						\State $x \gets x'$ \label{step:x-gets-x'}
						\If{$x$ is integral} \Return $x$ \EndIf \label{step:return}
						\If {case \ref{event:round} happens for some $e = (u, v) \in E$}
							\State $x_e \gets 1, 
							B_v \gets x(\delta(v)), B_u \gets x(\delta(u))$ \label{step:round-e}
						\Else \Comment{case \ref{event:relax} happens for some $v$}
							\State $B_v \gets \infty$ \label{step:relax}
						\EndIf
					\EndWhile
				\end{algorithmic}
			\caption{Iterative Rounding Algorithm for Network Design.}
			\label{alg:iterative-rounding}
		\end{subfigure}
		\caption{The function $f$ and the iterative rounding algorithm.}
	\end{figure}

\subsection{The Iterative Rounding Algorithm}

Our iterative rounding algorithm is described in Figure~\eqref{alg:iterative-rounding}.  In Step~\ref{step:solve-LP}, we solve the convex relaxation~\eqref{LP} to obtain an extreme solution $x$, which can be done in polynomial time using standard techniques.   Then in Step~\ref{step:init-B} we define $B_v = \max\{x(\delta(v)), 1\}$ for every $v$ to be the upper bound on the degree of $v$. So, before Loop~\ref{step:while}, we have $x \in \calP \cap \calQ_B$. We shall maintain this property before and after each iteration of the loop. 
	
	In each iteration of Loop~\ref{step:while}, we randomly choose a vertex point $x'$ of $\calP \cap \calQ_B$ such that $\E[x'] = x$ (Step~\ref{step:x'}) and then update $x$ to be the $x'$ (Step~\ref{step:x-gets-x'}).  This is possible since at the beginning of the iteration we have $x \in \calP \cap \calQ_B$.   If $x$ is integral, we then return $x$ in Step \ref{step:return}. If we did not return, by that $\calP$ is $(\alpha, \beta)$-integral, either \ref{event:round} or \ref{event:relax} happens. In the former case,  we update $x_e$ to $1$, and change $B_v$ and $B_u$ for the two end vertices $u, v$ of $e$ so that we still have $B_{v'} = \max\{x(\delta(v')), 1\}$ for every $v' \in V$ (Step~\ref{step:round-e}). In the latter case, we change $B_v$ to $\infty$ so that there will be no degree constraint for $v$ from now on.  Notice that in either case, we maintain the invariant that $x \in \calP \cap \calQ_B$ as $\calP$ is upward-closed.  

	Notice that once $x_e$ becomes $0$ or $1$ in some iteration, it will remain unchanged. This holds since for $\E[x'_e] = x_e \in \{0, 1\}$ to hold, we must always have $x'_e = x_e$. When the algorithm terminates, it returns an integral $x$ which satisfies the connectivity requirements.  This holds since we have $x \in \calP$ and $\calP$ is good.  The algorithm will terminate in $O(|E|)$ iterations since in every iteration, we either fixed the value of some $x_e$ to 1, or changed some $B_v$ from a finite number to $\infty$.

\subsection{Analysis of the Algorithm}

We now begin to analyze the algorithm.  As discussed, the algorithm will terminate with a subgraph which satisfies the connectivity requirements.  
To prove Theorem~\ref{thm:polytope-alg}, we need to analyze the total cost and the $\ell_p$-norm of the degrees.  

In Step~\ref{step:round-e}, we say that we \emph{round} the edge $e$.  In Step~\ref{step:relax}, we say we \emph{relax} the vertex $v$. 
	At any time of the algorithm, we define a vector $\bar x \in [0, 1]^E$ as follows. If $e$ has not been rounded yet, then let $\bar x_e = x_e$.  Otherwise, let $\bar x_e$ be the value of $x_e$ right before Step~\ref{step:round-e} in which we round $e$. Thus, from the moment, $\bar x_e$ remains unchanged.
		
	Let $T$ be the number of iterations we run Loop~\ref{step:while}; notice that this is a random variable.   For every integer $t \in [0, T]$, we let $x^t, \bar x^t, B^t$ to be the values of $x, \bar x, B$ at the end of the $t$-th iteration of the Loop~\ref{step:while}. So $x^T$ is the output of the algorithm. 

\begin{obs}
    \label{obs:barx}
	The following statements are true. 
	\begin{enumerate}[leftmargin=*, label=(\ref{obs:barx}\alph*), itemsep=0pt]
		\item During Loop~\ref{step:while}, we always have $\bar x_e \leq x_e \leq \alpha \bar x_e$ for every $e \in E$.  \label{property:barx-to-x}
		\item Assume $x^0(\delta(v)) < 1$ for a vertex $v \in V$. Then at the first moment when  $x(\delta(v)) \geq 1$ holds, we have $\bar x(\delta(v)) \leq 1$. \label{property:delta-v-becomes-1}
		\item $\bar x(\delta(v))$ does not change from the first moment $x(\delta(v)) \geq 1$ holds, until the moment $v$ is relaxed, or the end of the algorithm if this does not happen. \label{property:barx-no-change}
	\end{enumerate}
\end{obs}

\begin{proof}
    $\bar x_e \leq x_e$ by the definition of $\bar x_e$.  Moreover, $x_e \leq \alpha \bar x_e$ as if $x_e > \bar x_e$, then $x_e = 1$ and $x_e \geq \frac1\alpha$. So \ref{property:barx-to-x} holds.

    To prove \ref{property:delta-v-becomes-1}, we consider two scenarios. In the first scenario, the moment is after Step~\ref{step:x-gets-x'} in some iteration.  In this scenario, $\bar x(\delta(v)) = x(\delta(v)) = 1$ since $B_v = 1$ at the moment. In the second scenario, the moment is after we round some edge $e \in \delta(v)$ in Step~\ref{step:round-e}.  In this case $\bar x(\delta(v))$ is the same as $x(\delta(v))$ before the step, which is strictly less than 1.
	
    \ref{property:barx-no-change} holds since we maintained $B_v = x(\delta(v))$ from the moment $x(\delta(v))$ becomes at least 1. If $\bar x_e \neq x_e$ at some time, it must be the case that $x_e = 1$.  In this case, both $x_e$ and $\bar x_e$ will not change in the future. 
\end{proof}

We can now analyze the expected cost of the algorithm.  First, though, we will need a structural result.
  
    \begin{lemma}
        \label{lemma:martingale}
        For every edge $e \in E$, the sequence ${\bar x}^0_e, {\bar x}^1_e, \cdots, {\bar x}^T_e$ is a martingale.
    \end{lemma}
	\begin{proof}
		Focus on an iteration $t \geq 1$ and edge $e \in E$, and we fix the sequence ${\bar x}^0_e,  {\bar x}^1_e, \cdots, {\bar x}^{t-1}_e$. For simplicity we use $\E'[\cdot]$ to denote $\E[\cdot | {\bar x}^0_e,  {\bar x}^1_e, \cdots, {\bar x}^{t-1}_e]$. We need to prove $\E'[\bar x^t_e] = \bar x^{t-1}_e$.
		
		If we rounded $e$ in iteration $t$ or before, then $\bar x^t_e = \bar x^{t-1}_e$ happens with probability 1. So, we can assume that $e$ has not been rounded by the end of iteration $t$.  In this case, $\bar x^{t-1}_e = x^{t-1}_e$.
		
		So, in iteration $t$, either \ref{event:round} happens for some $e' \neq e$, or \ref{event:relax} happens. In either case, we have $\E'[{\bar x}^t_e] = \E'[x^t_e]  = x^{t-1}_e = {\bar x}^{t-1}_e$ by the way we define the distribution for $x'$ in Step~\ref{step:x'}. Therefore, ${\bar x}^0_e, {\bar x}^1_e, \cdots, \bar x^T_e$ is a martingale. 
	\end{proof}
	
	\begin{corollary} \label{cor:cost}
		$\E\left[\sum_{e \in E}c_ex^T_e\right] \leq \alpha \sum_{e \in E} c_e x^0_e$.
	\end{corollary}
	\begin{proof}
		\begin{align*}
			\E\left[\sum_{e \in E}c_ex^T_e\right] \leq \alpha \E\left[\sum_{e \in E}c_e\bar x^T_e\right] = \alpha \sum_{e \in E}c_e \bar x^0_e = \alpha \sum_{e \in E} c_e x^0_e.
		\end{align*}
		The inequality is by \ref{property:barx-to-x} and the first equality used Lemma~\ref{lemma:martingale}.
	\end{proof}
	
Now that we understand the expected cost, it only remains to analyze the degree constraint.  From now on we fix a vertex $v \in V$. We upper bound $x^T(\delta(v))$, which will in turn give an upper bound on $\E\left[(x^T(\delta(v)))^p\right]$. The main lemma we prove is 
    \begin{lemma}
        \label{lemma:bound-degree}
        For every $v \in V$, we have $\E\left[(x^T)^p(\delta(v))\right] \leq \alpha (\alpha + \beta)^{p-1} \cdot f(x^0(\delta(v)))$.
    \end{lemma}
    \begin{proof}
	We first consider the case $x^0(\delta(v)) \geq 1$. Let $t$ be the iteration in which $v$ is relaxed, or let $t = T$ if $v$ is not relaxed during the algorithm. By Property~\ref{property:barx-no-change}, $\bar x(\delta(v))$ does not change until the end of iteration $t$. Then, we have $x^T(\delta(v)) \leq  x^t(\delta(v)) + \beta \leq \alpha \bar x^t(\delta(v)) + \beta = \alpha \bar x^0(\delta(v)) + \beta = \alpha x^0(\delta(v)) + \beta$. Notice that this happens with probability 1.

	Notice that $\E[x^T(\delta(v))] \leq \alpha \E [\bar x^T(\delta(v))] = \alpha \bar x^0(\delta(v)) = \alpha x^0(\delta(v))$ by Lemma~\ref{lemma:martingale}. We have:
		\begin{align*}
			\E\left[(x^T(\delta(v)))^p\right]  \leq  \frac{\alpha x^0(\delta(v))}{\alpha x^0(\delta(v)) + \beta} (\alpha x^0(\delta(v)) + \beta)^p = \alpha x^0(\delta(v))(\alpha x^0(\delta(v)) + \beta)^{p-1}.
		\end{align*}
		This implies 
		\begin{align*}
			\frac{\E\left[(x^T)^p(\delta(v))\right]}{f(x^0(\delta(v)))} \leq \frac{\alpha x^0(\delta(v))(\alpha x^0(\delta(v)) + \beta)^{p-1}}{(x^0(\delta(v)))^p} = \alpha \left(\alpha + \frac{\beta}{x^0(\delta(v))}\right)^{p-1} 
			\leq \alpha(\alpha+\beta)^{p-1}.
		\end{align*}
		
	Now we consider the second case: $x^0(\delta(v)) < 1$.  Assume $x(\delta(v)) \geq 1$ happens at some time of the algorithm. By \ref{property:delta-v-becomes-1}, at the first moment when $x(\delta(v)) \geq 1$, we have $\bar x(\delta(v)) \leq 1$. By \ref{property:barx-no-change}, from the moment until the moment $v$ becomes relaxed (or until the end of the algorithm if $v$ is never relaxed), $\bar x(\delta(v))$ does not change.  Therefore, immediately after $v$ becomes relaxed, we have $\bar x(\delta(v)) \leq 1$. Thus $x^T(\delta(v))$ is at most the value of $x(\delta(v)) + \beta$ at this moment, which is at most $\alpha\bar x(\delta(v)) + \beta \leq \alpha + \beta$.   Again, we have $\E[x^T(\delta(v))] \leq \alpha x^0(\delta(v))$. So 
	\begin{align*}
		\E\left[(x^T(\delta(v)))^p\right]  \leq  \frac{\alpha x^0(\delta(v))}{\alpha + \beta} (\alpha + \beta)^p = \alpha x^0(\delta(v))(\alpha  + \beta)^{p-1}.
	\end{align*}
	Then, 
	\begin{align*}
		\frac{\E\left[(x^T)^p(\delta(v))\right]}{f(x^0(\delta(v)))} \leq \frac{\alpha x^0(\delta(v))(\alpha + \beta)^{p-1}}{x^0(\delta(v))} = \alpha(\alpha+\beta)^{p-1}.
	\end{align*}
	So, we always have $\E\left[(x^T)^p(\delta(v))\right] \leq \alpha(\alpha+\beta)^{p-1} f(x^0(\delta(v)))$. This implies $\E[\sum_{v}(x^T_\delta(v))^p] \leq \alpha(\alpha+\beta)^{p-1} A^p$.

        Now consider the case where $x(\delta(v)) \geq 1$ never happens; that is, $x^T(\delta(v)) < 1$. As $\E\left[x^T(\delta(v))\right] = x^0(\delta(v))$. Then we have $\E\left[(x^T)^p(\delta(v))\right] \leq x^0(\delta(v))$. The lemma clearly holds.
    \end{proof}

Corollary~\ref{cor:cost} and Lemma~\ref{lemma:bound-degree} imply Theorem~\ref{thm:polytope-alg}, which in turn implies Theorem~\ref{thm:network-design}.

\section{A Tree Labeling Problem}
\label{sec:tree-labeling}
	In this section, we introduce a tree labeling problem to which we reduce the Group Steiner Tree problem with degree bounds on bounded-treewidth graphs. We are given a full binary tree $\bfT = (\bfV, \bfE)$ rooted at $\bfr \in \bfV$.\footnote{It is not important to require the binary tree to be full; our algorithm works when some internal node has only one child. Assuming every internal node have 2 children is only for notational convenience.}  For every vertex $u \in \bfV$, we are given a finite set $L_u$ of \emph{labels} for $u$; we assume $L_u$'s are disjoint and let $L:= \union_{u \in \bfV}L_u$.   The output is a labeling $\vec\ell = (\ell_u \in L_u)_{u \in \bfV}$ of  the vertices $\bfV$, that satisfies the constraints described below. 
	\begin{itemize}
		\item {\bf (consistency constraints)}\ \  For every internal node $u$ of $\bfT$ with two children $v$ and $v'$, we are given a set $\Gamma_u \subseteq L_u \times L_{v} \times L_{v'}$. A valid labeling $\vec \ell$ must satisfy $(\ell_u, \ell_v, \ell_{v'}) \in \Gamma_u$.
		\item {\bf (covering constraints)}\ \ We are given $k$ subsets $S_1, S_2, \cdots, S_k \subseteq L$. A valid labeling $\vec \ell$ needs to satisfy that for every $t \in [k]$, $\ell(\bfV) \cap S_t \neq \emptyset$, where $\ell(\bfV)$ is defined as $\{\ell_u: u \in \bfV\}$. In words, $\ell(\bfV)$ needs to intersect every $S_t$. 
		\item {\bf (cost constraints)}\ \  We are given $m \geq 0$ linear constraints defined by the costs $(c^i_\ell \in [0, 1])_{i \in [m], \ell \in L}$. For every $i \in [m]$, a valid labeling $\vec \ell$ needs to satisfy $\sum_{u \in \bfV}c^i_{\ell_u} \leq 1$. In words, there  are $m$ types of resource, and we have 1 unit of each type. Setting the label of $u$ to $\ell$ will use $c^i_\ell$ units of type $i$-resource.
	\end{itemize}

	We say a labeling $\vec\ell = (\ell_u \in L_u)_{u \in \bfV}$ is \emph{consistent} if it satisfies the consistency constraints.  Given a consistent labeling $\vec\ell$, we say it covers group $S_t$ if $\ell(\bfV) \cap S_t \neq \emptyset$. We define its type-$i$ cost to be $\cost^i(\vec \ell):= \sum_{u \in \bfV}c^i_{\ell_u}$.
	So a valid labeling $\vec \ell$ for the instance is a consistent one that covers all groups, and has $\cost^i(\vec\ell) \leq 1$ for every $i \in [m]$. 

	Given a label tree instance, we let $n = |\bfV|$,  $D$ be the height of $\bfT$ (the maximum number of edges in a root-to-leaf path in $\bfT$)  
 and $\Delta = \max_{u \in \bfV}|L_u|$ be the maximum size of any $L_u$.  The main theorem we prove is the following:
	\begin{theorem}
		\label{thm:label-tree-main}
		Assume we are given a feasible label tree instance $(\bfT = (\bfV, \bfE), \bfr, (L_u)_u, \break (\Gamma_u)_u, (S_t)_{t \in [k]}, A \in [0, 1]^{m \times L})$, i.e., there is a valid labeling. There is a randomized algorithm that in time $\poly(n) \cdot \Delta^{O(D)}$ outputs a consistent labeling $\vec \ell$ such that the following holds.
		\begin{enumerate}[label=(\ref{thm:label-tree-main}\alph*), leftmargin=*]
			\item For every $t \in [k]$, we have $\Pr[\vec \ell \text{ covers group }S_t]  \geq \frac{1}D$. \label{property:label-tree-covering}
			\item For every $i \in [m]$,  we have $\E\big[\exp\big(\ln(1 + \frac{1}{2D}) \cdot \cost^i(\vec \ell)\big)\big] \leq 1+\frac1D$.
			\label{property:label-tree-cost}
		\end{enumerate}
	\end{theorem}	
	Property~\ref{property:label-tree-cost} gives a tail concentration bound on $\cost^i(\vec \ell)$. The remaining part of this section is dedicated to the proof of Theorem~\ref{thm:label-tree-main}. 
 
\subsection{Construction of a  super-tree $T^\circ$}	
	In this section, we construct a rooted tree $T^\circ = (V^\circ, E^\circ)$ of size $O(n)\Delta^{O(D)}$ such that a consistent labeling of $\bfT$ corresponds to what we call \emph{a consistent sub-tree}. 
    So we can reduce the problem to finding the latter object. The root of $T^\circ$ is $r$.  Each internal node of $T^\circ$ is either a \emph{selector node}, or a \emph{copier node}; their meanings  will be clear soon. Each node $p \in V^\circ$ is \emph{associated with} a node $u$ in $T$. Each non-root selector node or leaf node is \emph{associated with} a label $\ell \in L_u$.  We shall use $p$ and $q$ and their variants to denote nodes in $T^\circ$, and $u$ and $v$ and their variants to denote nodes in $\bfT$. 
	
	The algorithm for constructing $T^\circ$ is described in Algorithm~\ref{alg:construct-T-circ}, which calls the procedure $\mathsf{construct\mhyphen tree}$ described in Algorithm~\ref{alg:construct-tree}. See Figure~\ref{fig:tree} for the illustration of the construction of $T^\circ$ from $\bfT$. For a node $p \in V^\circ$, we use $\Lambda(p)$ denotes the set of children of $p$ in $T^\circ$, and $\Lambda^*(p)$ denotes the set of descendants of $p$ in $T^\circ$, including $p$ itself.

	\begin{algorithm}[H]
		\caption{Main algorithm for the construction of $T^\circ$}
		\label{alg:construct-T-circ}
		\begin{algorithmic}[1]
			\State create a node $r$ associated with $\bfr$ as the root of $T^\circ$, and let $r$ be a \emph{selector node}
			\For{every $\ell \in L_{\bfr}$}: 
				\State create a child $p$ of $r$, associated with node $\bfr$ and label $\ell$
				\State call $\mathsf{construct\mhyphen tree}(p, \bfr, \ell)$
			\EndFor
		\end{algorithmic}
	\end{algorithm}

	\begin{algorithm}[H]
		\caption{$\mathsf{construct\mhyphen tree}(p, u, \ell)$ \Comment{$p \in V^\circ, u \in \bfV, \ell \in L_u$}}
		\label{alg:construct-tree}
		\begin{algorithmic}[1]
			\If{$u$ has no children} \Return \Comment{$p$ is a leaf node.}\EndIf 
			\State let $p$ be a \emph{selector node}, 
	let $v$ and $v'$ be the two children of $u$ in $\bfT$
			\For{every $\ell' \in L_v, \ell'' \in L_{v'}$ such that $(\ell, \ell', \ell'') \in \Gamma_u$}
				\State create a child $p'$ of $p$, associated with $u$, let $p'$ be a \emph{copier node}, 
                    \State create two children $q$ and $q'$ of $p'$, associate $q$ with node $v$ and label $\ell'$, associate $q'$ with node $v'$ and label $\ell''$ 
				\State call $\mathsf{construct\mhyphen tree}(q, v, \ell')$ and $\mathsf{construct\mhyphen tree}(q', v', \ell'')$
			\EndFor
		\end{algorithmic}
	\end{algorithm}

    \begin{figure}
        \centering
        \iflipics
        \includegraphics[scale=0.81]{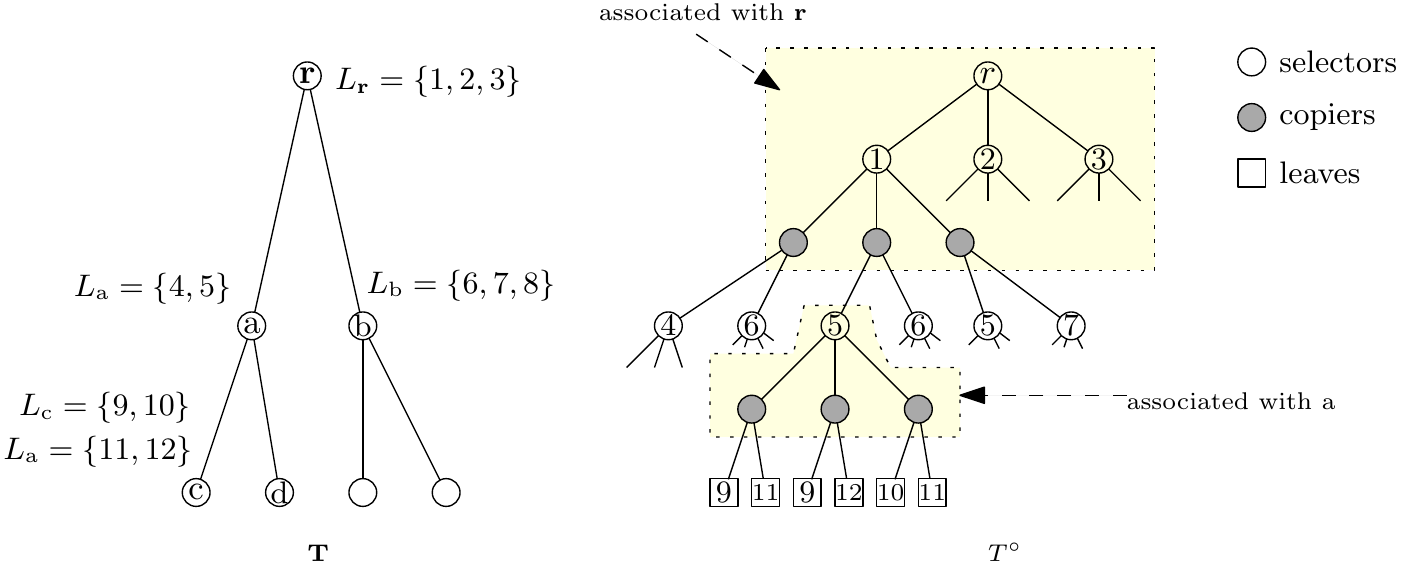}
        \else
        \includegraphics[scale=1]{tree.pdf}
        \fi
        \caption{An example for the construction of $T^\circ$. The tree on the left side is $\bfT$, and the tree on the right side is $\bfT^\circ$. The labels of the nodes in $\bfT$ are shown besides them. In $T^\circ$, selectors, copiers and leaves are denoted as empty circles, solid circles and empty squares respectively. The nodes in the two yellow polygons are associated with $\bfr$ and a respectively. The numbers in the circles and squares indicate the labels associated with the nodes. In the example, the triples in $\Gamma_\bfr$ with the first coordinate being $1$ are $(1,4, 6), (1, 5, 6)$ and $(1, 5, 7)$. The triples in $\Gamma_{\mathrm{a}}$ with the first coordinate being $5$ are $(5, 9, 11), (5, 9, 12)$ and $(5, 10, 11)$.}
        \label{fig:tree}
    \end{figure}
	
    Now we can define consistent sub-trees of $T^\circ$:
	\begin{definition}[Consistent sub-trees]
	Given a sub-tree $T$ of $T^\circ$ that contains $r$, we say $T$ is \emph{consistent} if the following conditions hold.
	\begin{itemize}[itemsep=0pt]
		\item Every selector node $p$ in $T$ has exactly one child in $T$.
		\item If $p$ is a copier node in $T$, then both of its children in $T^\circ$ are in $T$. 
	\end{itemize}
\end{definition}
The definition explains the names ``selector'' and ``copier'': a selector node $p$ in $T$ needs to select one of its children in $T^\circ$ and add it to $T$, and the children of a copier node $p$ will follow the node $p$ to enter $T$.  

It is easy to see a one-to-one correspondence between consistent labelings $\vec \ell = (\ell_u \in L_u)_{u \in \bfV}$ of $\bfT$, and consistent sub-trees $T$ of $T^\circ$.  Given the consistent labeling $\vec \ell$, the correspondent sub-tree $T$ of $T^\circ$ can be constructed as follows. First, we add $r$ and its child $p$ associated with label $\ell_\bfr$ to $T$. Then we grow the tree from $p$ using a recursive procedure.  Assume $p$ is associated with node $u$ in $\bfT$ and label $\ell \in L_u$. If $u$ is a leaf, we stop the procedure. Otherwise let $v$ and $v'$ be the two children of $u$, then we add the copier child $p'$ of $p$ that corresponds to the tuple $(\ell_\bfr, \ell_v, \ell_{v'})$ to $T$. We also add its two children $q$ and $q'$ to $T$. Then we run the procedure recursively over $q$ and $q'$.  Conversely, given a consistent sub-tree $T$ of $T^\circ$, we can recover a consistent labeling $\vec \ell$ of $\bfT$. 

For convenience, we extend the costs $(c^i_\ell)_{i \in [m], \ell \in L}$ to vertices in $V^\circ$: For every non-root selector node or leaf node  $p \in V^\circ$ associated with a label $\ell$, we define $c^i_p = c^i_\ell$ for every $i \in [m]$. For the root or a copier node $p$, we define $c^i_p = 0$.  For a consistent sub-tree $T = (V, E)$ of $T^\circ$, and $i \in [m]$, we define its type-$i$ cost to be $\cost^i(T) = \sum_{p \in V}c^i_p$.
This will be the same as $\cost^i(\vec \ell)$, for the labeling $\vec \ell$ correspondent to $T$.

We also extend the groups $S_1, S_2, \cdots, S_k$ to node sets in $T^\circ$: for every $t \in [k]$,  $S'_t$ contains the set of nodes $p \in V^\circ$ whose associated label is in $S_t$.  Then, a consistent labeling $\vec \ell$ covers a group $S_t$ if and only if the correspondent sub-tree $T = (V, E)$ covers $S'_t$, namely, $V \cap S'_t \neq \emptyset$.

Therefore, we are guaranteed that there is a consistent sub-tree $T^*$ of $T^\circ$ that covers all groups $S'_1, S'_2, \cdots, S'_k$, and has $\cost^i(T^*) \leq 1$ for every $i \in [m]$. Our goal is to output a random consistent sub-tree $T$ satisfying the conditions correspondent to \ref{property:label-tree-covering} and \ref{property:label-tree-cost}. This is done using an LP-based algorithm. 

\subsection{The LP relaxation for finding $T = (V, E)$}
Now we describe the LP relaxation that we use to find $T = (V, E)$.  For every vertex $p \in V^\circ$, we use $x_p$ to indicate if $p$ is in $T$, i.e., $p \in V$.   For every $t \in [k]$ and $q \in S'_t$, we use $y^t_q$ to indicate if $q$ is the node in $T$ we choose to cover $S'_t$.  There might be multiple nodes in $V \cap S'_t$, and in this case, we only choose one node in the set to cover $S'_t$; the choice can be arbitrary.  The LP is as follows.

\noindent\begin{minipage}[t]{0.54\textwidth}
	\begin{alignat}{2}
		x_r &= 1\label{LPC:root}\\
		\sum_{q \in \Lambda(p)} x_q  &= x_p &\qquad &\forall \text{ selector } p \in {V^\circ} \label{LPC:selector}\\
		x_q &= x_p &\qquad &\forall \text{ copier } p \in {V^\circ}, q \in \Lambda(p) \label{LPC:copier}\\
		x_p &\geq 0 &\qquad &\forall p \in V^\circ \label{LPC:non-negative}
	\end{alignat}
\end{minipage}\hfill
\begin{minipage}[t]{0.44\textwidth}
	\begin{alignat}{2}
		\sum_{q \in \Lambda^*(p) \cap S'_t}y^t_q &\leq x_p &\qquad &\forall p \in V^\circ, t \in [k]\label{LPC:group-packing}\\
		\sum_{q \in S'_t}y^t_q &= 1 &\qquad &\forall t \in [k] \label{LPC:group-covering}\\
		\sum_{q \in \Lambda^*(p)} c^i_q x_q &\leq x_p &\qquad &\forall p \in V^\circ, i \in [m] \label{LPC:cost}
	\end{alignat}
\end{minipage}

Constraints \eqref{LPC:root}-\eqref{LPC:non-negative} in the LP are for the consistency requirements. \eqref{LPC:root} says the root is always in $T$. \eqref{LPC:selector} says if a selector node $p$ is in $T$, then exactly one of its children is in $T$. \eqref{LPC:copier} says if a copier node $p$ is in $T$, and $q$ is a child of $p$, then $q$ is also in $T$. \eqref{LPC:non-negative} is the non-negativity condition. \eqref{LPC:group-packing} and \eqref{LPC:group-covering} deal with the covering requirements. \eqref{LPC:group-packing} says if $p$ is in $T$, then we choose at most one descendant of $p$ to cover the group $S'_t$; notice that the constraint implies $y^t_p \leq x_p$ if $p \in S'_t$.  \eqref{LPC:group-covering} says we choose exactly one node in $T$ to cover $S'_t$. \eqref{LPC:cost} handles the cost requirement: If $p$ is included in $T$, then the type-$i$ cost of the descendants of $p$ in $T$ is at most $1$.

\subsection{The rounding algorithm}
We solve LP\eqref{LPC:root} to obtain a solution $x \in [0, 1]^{V^\circ}$.  We add $r$ to $T$ and call $\mathsf{recursive\mhyphen rounding}(r)$ to obtain a sub-tree $T = (V, E)$. The procedure is defined in Algorithm~\ref{alg:recursive-rounding}.

\begin{algorithm}
	\caption{$\mathsf{recursive\mhyphen rounding}(p)$}
	\label{alg:recursive-rounding}
	\begin{algorithmic}[1]
		\If{$p$ is a selector node}
			\State choose one vertex $q \in \Lambda(p)$ randomly, so that $q$ is chosen with probability $\frac{x_q}{x_p}$
			\State add $q$ to $T$, and call $\mathsf{recursive\mhyphen rounding}(q)$
		\Else \Comment{$p$ is a copier or leaf node}
			\For{every $q \in \Lambda(p)$}
				\State with probability $\frac{x_q}{x_p} = 1$: add $q$ to $T$, and call $\mathsf{recursive\mhyphen rounding}(q)$
			\EndFor
		\EndIf
	\end{algorithmic}
\end{algorithm}

\begin{obs}
	$T$ is always consistent. For every $p \in V^\circ$, we have $\Pr[p \in V] = x_p$.
\end{obs}
\begin{proof}
	For a selector node $p$ in $T$, we always choose exactly one child of $p$ and add it to $T$.  For a copier node $p$ added to $T$ and one of its child $q$, $q$ is added to $T$ with probability 1. By the probabilities we add nodes to $T$, we can see that $\Pr[p \in V] = x_p$ for every $p \in V^\circ$. 
\end{proof}

\subsection{Analysis of probabilities of group coverage}
In this section, we fix $t \in [k]$ and analyze the probability that $T$ coves the group $S'_t$; or equivalently, the correspondent labeling covers the group $S_t$.  This will prove Property~\ref{property:label-tree-covering}. For every vertex $p \in V^\circ$, we define $z_p := \sum_{q \in \Lambda^*(p) \cap S'_t}y^t_q$, which indicates whether $S'_t$ is covered by vertices in the sub-tree rooted at $p$. By \eqref{LPC:group-packing}, we have $z_p \leq x_p$. By \eqref{LPC:group-covering}, we have $z_r = 1 = x_r$.  

We define the height of a node $p \in V^\circ$ to be the maximum number of copier nodes in a path from $p$ to one of its descendant leaves. We bound the probability that the tree rooted at $p$ covers $S'_t$ using inductions:
\begin{lemma}
    Assume $p \in V^\circ$ has height $h$. Then we have $\Pr\Big[\Lambda^*(p) \cap V \cap  S'_t \neq \emptyset \big| p \in V\Big] \geq \frac1{h+1}\frac{z_p}{x_p}$. 
\end{lemma}

\begin{proof}
	If $p \in S'_t$ then $\Pr\Big[\Lambda^*(p) \cap V \cap  S'_t \neq \emptyset \big|{p \in V}\Big] = 1 \geq \frac{z_p}{x_p}$. The inequality holds trivially. So, we can assume $p \notin S'_t$, and we prove the lemma for nodes $p$ from bottom to top in the tree $T^\circ$.  Suppose $p$ is a leaf; then $h = 0$, and $z_p = 0$ as we assumed $p \notin S'_t$. The inequality trivially holds. 

	So we can assume $p$ be a non-leaf node of height $h$, and assume the lemma holds for every $q \in \Lambda(p)$.  First assume $p$ is a selector node.  Then all children of $p$ have height at most $h$. 
	\begin{align*}
		\Pr\Big[\Lambda^*(p) \cap V \cap  S'_t \neq \emptyset \big| p \in V\Big] &\geq \sum_{q \in \Lambda(p)} \frac{x_q}{x_p}\cdot \frac{1}{h+1}\cdot\frac{z_q}{x_q} =  \sum_{q \in \Lambda(p)}  \frac{1}{h+1}\cdot\frac{z_q}{x_p} = \frac1{h+1} \cdot \frac{z_p}{x_p}.
	\end{align*}
	
	Then consider the case that $p$ is a copier node. All children of $p$ have height at most $h-1$.  Even though $p$ has exactly two children, our analysis works if it has any number of children. 
	\begin{align*}
		 \Pr&\Big[\Lambda^*(p) \cap V \cap  S'_t \neq \emptyset \big| p \in V\Big] \geq 1-\prod_{q \in \Lambda(p)}\left(1 - \frac{1}{h} \cdot \frac{z_q}{x_q}\right) = 1 - \prod_{q \in \Lambda(p)}\exp\left(-\frac{1}{h} \cdot \frac{z_q}{x_p}\right)\\ 
        &= 1 - \exp\left(-\frac{1}{h}\cdot  \frac{z_p}{x_p}\right) \geq \frac{1}{h} \cdot \frac{z_p}{x_p} - \frac12\left(\frac{1}{h} \cdot \frac{z_p}{x_p}\right)^2 \geq \frac{1}{h}\cdot \frac{z_p}{x_p} - \frac12\left(\frac{1}{h} \right)^2\frac{z_p}{x_p} \\
        &= \left(\frac{2h-1}{2h^2}\right)\frac{z_p}{x_p} \geq \frac{1}{h+1}\cdot \frac{z_p}{x_p}.
	\end{align*}
	The first equality in the first line used that $x_q = x_p$ for every $q \in \Lambda(p)$. The second equality used that $z_p = \sum_{q \in \Lambda(p)}z_q$ as $p \notin S'_t$. The first inequality in the second line used that $e^{-\theta} \leq 1-\theta + \frac {\theta^2}2$ for every $\theta \geq 0$.   The second inequality used that $\frac{z_p}{x_p} \leq 1$.
\end{proof}
Notice that the height of the root $r$ of $T^\circ$ is $D-1$. Applying the above lemma with $p = r$, we have that $T$ covers group $S'_t$ with probability at least $\frac{1}{D} \cdot \frac{z_r}{x_r} = \frac{1}{D}$. So, the correspondent $\vec \ell$ covers $S_t$ with probability at least $\frac1D$, proving Property~\ref{property:label-tree-covering}.

\subsection{Concentration bound on costs}
\label{appendix:concentration}
In this section, we prove Property~\ref{property:label-tree-cost}. To this end, we fix an index $i \in [m]$ and analyze the type-$i$ cost of $T = (V, E)$. For notation convenience, we use $c_p$ to denote $c^i_p$, and cost for type-$i$ cost.

For every vertex $p \in V^\circ$, let $w_p = \sum_{q \in \Lambda^*(p)} c_q x_q$ be the fractional cost incurred by the sub-tree of $T^\circ$ rooted at $p$.  By \eqref{LPC:cost}, we have $w_p \leq x_p$.  Let $W_p = \sum_{q \in \Lambda^*(p) \cap V} c_q$ be the cost of $T$ incurred by descendants of $p$. So, we have  $\E[W_p] = w_p$.

As is typical, we shall introduce a parameter $s > 0$ and consider the expectation of the random exponential variables $\bfe^{s W_p}$.  Later we shall set $s = \ln(1+\frac1{2D})$, but the main lemma holds for any $s > 0$. We define an $\alpha_h$ for every integer $h \geq 0$ as 
$\alpha_0 = \bfe^s$ and $\alpha_h = \bfe^{\alpha_{h-1}-1}, \forall h \geq 1$.
Notice that $\alpha_0, \alpha_1, \ldots$ is an increasing sequence.

In this section, we count selector nodes in the definition of heights: the height of a node $p \in V^\circ$ is the maximum number of selector nodes in a path from $p$ to its descendant leaf.   The main lemma we prove in this section is:
	\begin{lemma}
		\label{lemma:bound-exp-mp}
		For any node $p$ in $T^\circ$ of height $h$, we have 
		$
		\E\Big[\bfe^{s W_p} \big| p \in V \Big] \leq \alpha_h^{w_p/x_p}.
		$
	 \end{lemma}
 
\begin{proof}
	Again, we prove the lemma for nodes $p$ from bottom to top of the tree $T^\circ$.  Focus on a node $p$ of height $h$.  Consider the case where $p$ is a copier or leaf node. Then all children of $p$ has height at most $h$. 
	\begin{align*}
		\E\Big[\bfe^{sW_p}\big|p \in V\Big]  &= \bfe^{sc_p} \prod_{q \in \Lambda(p)}\E \Big[\bfe^{sW_q} \big| q \in V \Big] =\alpha_0^{c_px_p/x_p} \prod_{q \in \Lambda(p)}\alpha_h^{w_q/x_p} \leq \alpha_h^{c_px_p/x_p} \prod_{q \in \Lambda(p)}\alpha_h^{w_q/x_p}=\alpha_h^{w_p/x_p}.
	\end{align*}
	The last inequality used that $\alpha_0 \leq \alpha_h$, and the last equality used that $w_p = c_p x_p + \sum_{q \in \Lambda(p)} w_q$.

	Now suppose $p$ is a selector. Then all children of $p$ have height at most $h-1$.  Conditioned on $p \in V$, the rounding procedure adds exactly one child $q$ of $p$ to $V$.  Then, we have 
	\begin{align*}
		\E\Big[\bfe^{sW_p}\big|p \in V\Big]  &= \bfe^{sc_p}\cdot \sum_{q \in \Lambda(p)}\frac{x_q}{x_p} \E\Big[\bfe^{sW_q}\big|q \in V\Big] 
		= \bfe^{sc_p} \cdot \sum_{q \in \Lambda(p)}\frac{x_q}{x_p}\alpha_{h-1}^{w_q/x_q}\\
		&\leq \bfe^{sc_p} \left(\left(\frac{w_p}{x_p} - c_p\right)\cdot \alpha_{h-1} + \left(1-\frac{w_p}{x_p} + c_p\right)\right) = \bfe^{sc_p}\left(1 + \left(\frac{w_q}{x_q} - c_p\right)(\alpha_{h-1} - 1)\right)\\
		&\leq \bfe^{sc_p}\cdot \exp\left(\left(\frac{w_p}{x_p} - c_p\right)(\alpha_{h-1} - 1)\right) = \bfe^{sc_p} \cdot \alpha_h^{w_p/x_p - c_p} \leq \alpha_h^{w_p/x_p}.
	\end{align*}
	
	To see  the inequality in the second line, we notice the following four facts: (i) $\alpha_{h-1}^\theta$ is a convex function of $\theta$, (ii) $w_q/x_q \in [0, 1]$ for every $q \in \Lambda(p)$,  (iii) $\sum_{q \in \Lambda(p)}\frac{x_q}{x_p} = 1$ and (iv) $\sum_{q\in \Lambda(p)}\frac{x_q}{x_p}\cdot\frac{w_q}{x_q} = \sum_{q\in \Lambda(p)}\frac{w_q}{x_p} = \frac{w_p}{x_p} - c_p$.  The equality in the last line is by the definition of $\alpha_h$.   The last inequality used that $\bfe^s = \alpha_0 \leq \alpha_h$.
\end{proof}

The height of the root $r$ is $D$.\footnote{The height of $r$ is $D+1$ by definition, but Lemma~\ref{lemma:bound-exp-mp} holds when we define its height to be $D$, as one can collapse the first two levels of $T^\circ$ into one level.}  Now, we set $s = \ln(1+\frac{1}{2D})$. We prove inductively the following lemma:
\begin{lemma}
    \label{lemma:bound-alpha}
    For every $h \in[0, D]$, we have $\alpha_h \leq 1 + \frac{1}{2D - h}$.
\end{lemma}
\begin{proof}
	By definition, $\alpha_0 = \bfe^s = 1+ \frac{1}{2D}$ and thus the statement holds for $h = 0$.  Let $h \in [1, D]$ and assume the statement holds for $h-1$. Then, we have 
	\begin{align*}
		\alpha_h &= \bfe^{\alpha_{h-1}-1} \leq \bfe^{1 + \frac{1}{2D-h+1}} \leq 1 + \frac{1}{2D-h+1} + \left(\frac{1}{2D-h+1}\right)^2\\
		&= 1 + \frac{2D-h+2}{(2D-h+1)^2} \leq 1 + \frac{1}{2D - h}.
	\end{align*}
	The first inequality used the induction hypothesis and the second one used that for every $\theta \in [0, 1]$, we have $e^\theta \leq 1 + \theta + \theta^2$. 
\end{proof}

To wrap up, we apply Lemma~\ref{lemma:bound-exp-mp} on $p = r$. Notice that $r \in V$ always happens, at $W_r = \cost^i(T)$. We have $\E\big[\exp\big(\ln\big(1+\frac1{2D}\big) \cdot \cost^i(T)\big)\big] \leq \alpha_D^{w_r/x_r} \leq 1 + \frac1D$ by Lemma~\ref{lemma:bound-alpha} and that $w_r \leq x_r = 1$.  Using the correspondence between sub-trees of $T^\circ$ and labelings of $\bfT$ proves Property~\ref{property:label-tree-cost}.

\section{Reduction of Degree-Bounded Group Steiner Tree on Bounded-Treewidth Graphs to Tree-Labeling Problem} \label{sec:reduction}

In this section we prove Theorem~\ref{thm:GST-bd-treewidth}, by reducing Group Steiner Tree with degree bounds on bounded treewidth graphs to the tree labeling problem studied in Section~\ref{sec:tree-labeling}. Recall the input of the problem contains a graph $G = (V, E)$ with edge costs $c \in R_{\geq 0}^E$, a root $r$,  $k$ groups $S_1, S_2, \cdots, S_k$ of vertices, and a degree bound $\db_v \in \Z_{>0}$ for every $v \in V$.   
Without loss of generality, we assume $\{r\}, S_1, S_2, \cdots, S_k$ are mutually disjoint. Again, we use $\opt$ to denote the minimum-cost of a valid subgraph $H$.

Let $\bfT = (B, \bfE)$ be the tree decomposition of the graph $G = (V, E)$.  Every $b \in B$ is called a bag and let $X_b \subseteq V$ be the set of vertices contained in the bag $b$. We can add the root $r$ to all the bags, which increases the maximum size of a bag by at most 1. It was show in \cite{Bod88} that we can assume $\bfT$ is a rooted binary tree of depth $O(\log n)$, by sacrificing the bag size by an $O(1)$ factor. We summarize the properties as follows:
\begin{itemize}
	\item $\bfT$ is a full binary tree rooted at $\bfr$, with depth $O(\log n)$.
	\item $|X_b| \leq O(1)\cdot \tw$ for every $b \in B$.
	\item For every edge $(u, v) \in E$, there is some $b \in B$ with $\{u, v\} \subseteq X_ b$. 
	\item For every $v \in V$, the set of bags $b$ with $v \in X_b$ is connected in $\bfT$. 
\end{itemize}

For every $e \in E$, let $b_e$ be the highest node $b$ such that $X_b$ contains both end vertices of $e$. This is well-defined due to the last property in the above list.  For every $b \in B$, we let $E_b = \{e \in E: b_e = b\}$. So, $(E_b)_{b \in B}$ forms a partition of $E$. 

\subparagraph{Notations on Partitions.}
Given two partitions $\Pi$ and $\Pi'$ of a common set $X$,  we say $\Pi'$ refines $\Pi$ if any two elements in $X$ that are in the same set in $\Pi'$ are also in the same set in $\Pi$.  We use $\Pi' \leq \Pi$ to denote that $\Pi'$ refines $\Pi$.  Given two partitions $\Pi$ and $\Pi'$ of $X$, we use $\Pi \vee \Pi'$ to denote the join of $\Pi$ and $\Pi'$ w.r.t the relation $\leq$.  That is, we define a graph where there is an edge between $u$ and $v$ if they are in the same set in $\Pi$ or $\Pi'$. Then two vertices $u$ and $v$ are in the same set in the partition $\Pi \vee \Pi'$ if and only if they are in the same connected component in the graph. 

Abusing notations slightly, if an element $v$ is not included in a partition $\Pi$, we treat $\{v\}$ as a singleton set in $\Pi$. This allows us to extend the operators $\leq$ and $\vee$ to two partitions $\Pi$ and $\Pi'$ with different ground sets.   Given a partition $\Pi$ and a set $X$, we let $\Pi[X]$ be the partition $\Pi$ restricted to the ground set $X$: two elements $u, v \in X$ are in the same set in $\Pi[X]$ if and only if they are in the same set in $\Pi$.

For any set $F \subseteq E$ of edges, we define $\CC(F)$ to be the partition of the vertices incident to $F$, such that $u$ and $v$ are in the same set in $\CC(F)$ if and only if they are in the same connected component in $(V, F)$.

\subparagraph{Construction of Labels and Consistency Triples.}
The tree $\bfT$ for the tree-labeling instance is the same as the decomposition tree $\bfT$. (This is the reason we use the same notion $\bfT$.) So we have $\bfV = B$. Now we fix a bag $b \in B$ and define the set $L_b$ of labels for $b$.  
To define the labels, we let $H = (V_H, E_H)$ be any sub-graph of $G$, which we should think of as the output of the GST problem. Fix a bag $b \in B$, let $\Lambda^*(b)$ be the set of descendants of $b$ in $\bfT$, including $b$ itself. We then make the following definitions: 
\begin{itemize}
	\item $F_b(H) := E_H \cap E_b$ is the set of edges from $E_b$ that are included in $H$. 
	\item $\Pi^\downarrow_b(H)$ is the partition of $X_b$ so that two vertices $u,  v \in X_b$ is in the set in $\Pi^\downarrow_b(H)$ if and only if they are connected in the graph $(V_H, E_H \cap \bigcup_{b' \in \Lambda^*(b)} E_{b'})$.
	\item $\Pi^\uparrow_b(H)$ is the partition of $X_b$ so that two vertices $u,  v \in X_b$ is in the set in $\Pi^\downarrow_b(H)$ if and only if they are connected in the graph $(V_H, E_H \cap \bigcup_{b' \in B \setminus \Lambda^*(b) \cup \{b\}} E_{b'})$.
\end{itemize}
In words, $\Pi^\downarrow_b(H)$ and $\Pi^\downarrow_b(H)$ respectively indicate the partition of $X_b$ correspondent to the edges of $H$ in bags below and above $b$ respectively.  

Without knowing $H$, we can define the label set $L_b$ for $b$ to be all tuples $(F_b, \Pi^\downarrow_b, \Pi^\uparrow_b)$ such that $(F_b, \Pi^\downarrow_b, \Pi^\uparrow_b) = (F_b(H), \Pi^\downarrow_b(H), \Pi^\uparrow_b(H))$ for some valid output graph $H$.  We then define the consistency tuples $\Gamma_b$'s so that a consistent labeling gives a valid outputs sub-graph $H$.  \medskip

Formally, let $L_b$ be the set of all tuples $(F_b, \Pi^\downarrow_b, \Pi^\uparrow_b)$ such that 
\begin{itemize}
	\item $F_b \subseteq E_b$ is a forest over $X_b$, $\CC(F_b) \leq \Pi^\downarrow_b$  and $\CC(F_b) \leq \Pi^\uparrow_b$,
	\item if $b = \bfr$, then $\Pi^\uparrow_b = \CC(F_b)$, and 
	\item if $b$ is a leaf, then $\Pi^\downarrow_b = \CC(F_b)$. 
\end{itemize}

Then we define the set $\Gamma_b$ of triples,  for an inner vertex $b$ in $\bfT$ with two children $b'$ and $b''$.  We have $\big((F_b, \Pi^\downarrow_b, \Pi^\uparrow_b), (F_{b'}, \Pi^\downarrow_{b'}, \Pi^\uparrow_{b'}), (F_{b''}, \Pi^\downarrow_{b''}, \Pi^\uparrow_{b''})\big) \in \Gamma_b$ if and only if 
\begin{itemize}
	\item $\Pi^\downarrow_b = \Big(\Pi^\downarrow_{b'} \vee \Pi^\downarrow_{b''} \vee \CC(F_b)\Big)[X_b]$,
	\item $\Pi^\uparrow_{b'} = \Big(\Pi^\uparrow_{b} \vee \Pi^\downarrow_{b''}  \vee \CC(F_{b'})\Big)[X_b]$, and
	\item $\Pi^\uparrow_{b''} = \Big(\Pi^\uparrow_{b} \vee \Pi^\downarrow_{b'}  \vee \CC(F_{b''})\Big)[X_{b''}]$.
\end{itemize}

\begin{claim}\label{claim:connectivity-truthful}
	Let $\{(F_b, \Pi^\downarrow_b, \Pi^\uparrow_b)\}_{b \in B}$ be a consistent labeling of the tree $\bfT$.  Let $H = (V, \union_{b \in B} F_b)$.  Then we have $\Pi^\downarrow_b[H] = \Pi^\downarrow_b$ and $\Pi^\uparrow_b[H] = \Pi^\uparrow_b$ for every $b \in B$. 
\end{claim}
The claim says that if the labels are consistent, then $\Pi^\downarrow_b$ and $\Pi^\uparrow_b$ represent their true values.

\subparagraph{Construction of Covering and Cost Constraints.}  
The requirement that all groups are connected to $r$ can be captured by the covering constraint in the tree-labeling problem. For every $t \in [k]$, a label $(F_b, \Pi^\downarrow_b, \Pi^\uparrow_b) \in L_b$ for some $b \in B$ can satisfy the group $S_t$ if for some $s \in S_t$ we have $(s, r)$ are in the same set in the partition $\Pi^\downarrow_b \vee \Pi^\uparrow_b$.

The edge costs and degree constraints can be captured by the cost constraints in the tree-labeling instance. Consider the costs first.  Using binary search, we assume we know the optimum cost $C^*$ for the instance.  For every bag $b \in B$ and every label $(F_b, \Pi^\downarrow_b, \Pi^\uparrow_b)$, the cost of the label is $c(F_b) := \sum_{e \in F_b} c_e$. We disallow this label by removing it if $c(F_b) > C^*$. Scaling all costs by $C^*$ so that all costs are in $[0, 1]$.  So, the cost being at most $C^*$ in the group Steiner tree instance is equivalent to that the cost of all labels is at most $1$. 

Finally we consider the degree constraints $d_H(v) \leq \db_v$ for every $v \in V$. For every $v \in V$, we define a cost constraint in the tree-labeling instance. For every bag $b \in B$ with $v \in X_b$, and every label $(F_b, \Pi^\downarrow_b, \Pi^\uparrow_b)$, the cost of the label is $|\delta(v) \cap F_b|$, where $\delta(v)$ is the incident edges of $v$ in $G$. Again, we disallow the label if $|\delta(v) \cap F_b| > \db_v$, and we scale the costs by $\db_v$ so that all costs are in $[0, 1]$. Then the degree constraint on $v$ is reduced to this cost requirement in the tree labeling instance.

\subparagraph{Wrapping Up.} We then run the algorithm in Theorem~\ref{thm:label-tree-main} on the constructed tree-labeling instance.  Let $(F_b, \Pi^\downarrow_b, \Pi^\uparrow_b)$ be the label of a bag $b$, and let $H = (V, \union_{b \in B}F_b)$.  By Claim~\ref{claim:connectivity-truthful}, the consistency constraints guarantee that the $\Pi^\downarrow_b$ and $\Pi^\uparrow_b$ truthfully represent the connectivity of the graph $G$.  So, if the covering constraint for a group $S_t$ is satisfied, then $H$ indeed connects $r$ and $S_t$.  Recall that $D = O(\log n)$ is the depth of the tree $\bfT$. By Properties~\ref{property:label-tree-covering} and \ref{property:label-tree-cost}, we have 
\begin{itemize}
	\item For every $t \in [k]$, $H$ connects $r$ and $S_t$ with probability at least $\frac{1}{D}$.
	\item $\E\big[\exp(\ln (1+\frac1{2D}) \cdot \frac{c(H)}{C^*})\big] \leq 1 + \frac1D$.
	\item $\E\big[\exp(\ln (1+\frac1{2D}) \cdot \frac{d_H(v)}{\db_v})\big] \leq 1 + \frac1D$ for every $v \in V$. 
\end{itemize}

We run the algorithm for $M = \Theta(D\log n) = \Theta(\log^2 n)$ times, with a large hidden constant in the $O(\cdot)$ notation, and output the union $H$ of all sub-graphs constructed by the $M$ times.  With high probability, all groups are connected to $r$ in $H$.  $\E[\exp(\ln (1+\frac1{2D}) \cdot \frac{d_H(v)}{\db_v})] \leq (1 + \frac1D)^M = n^{O(1)}$.  Using Markov inequality, we have $\exp(\ln (1 + \frac1{2D})\cdot \frac{d_h(v)}{\db_v}) \leq n^{O(1)}$ for every $v \in V$ with high probability. That is, $d_h(v) \leq O(\db_v \log n \cdot D) = O(\log^2n )\db_v$ with high probability.  Similarly, with high probability, we have $c(H) \leq O(\log ^2n)C^*$.

We then analyze the running time of the algorithm. The key parameter deciding the running time is $\Delta$, the maximum size of a label set $L_b$.  As we assumed $F_b$ is a forest over $X_b$ and $|X_b| \leq O(\tw)$, there are $\tw^{O(\tw)}$ different possibilities for $F_b$.  There are also $\tw^{O(\tw)}$ possibilities for each of $\Pi^\downarrow_b$ and $\Pi^\uparrow_b$. So, $|L_b| \leq \tw^{O(\tw)}$ for every $b \in B$. Therefore,  the running time of the algorithm is $\poly(n) \cdot \Delta^{O(D)} = \poly(n) \cdot (\tw^{O(\tw)})^{O(\log n)} = n^{O(\tw \log \tw)}$. This finishes the proof of Theorem~\ref{thm:GST-bd-treewidth}.

\bibliographystyle{plainurl}
\bibliography{refs}

\end{document}